\documentclass[11pt]{article}
\usepackage[margin=1in]{geometry}

\usepackage{authblk}
\usepackage{hyperref}

\usepackage{xcolor,colortbl}
\usepackage{amsmath,amsfonts,amssymb,mathtools,amsthm}
\usepackage{multirow}

\usepackage{xspace}

\usepackage{subcaption}
\usepackage{algorithm}
\usepackage[noend]{algpseudocode}
\algnewcommand\Break{\textbf{break}}
\algnewcommand\Continue{\textbf{continue}}
\algnewcommand\TAnd{~\textbf{and}~}
\algnewcommand\True{\mathbf{true}}
\algnewcommand\False{\mathbf{false}}
\algnewcommand\Null{\mathbf{null}}

\newcommand{\R}{\mathbb{R}}

\newcommand{\ModSSBPshort}[0]{ASBPIC\xspace}
\newcommand{\ModSSBPlong}[0]{Arbitrary-Source Bottleneck Path with Initial Capacity\xspace}

\newcommand{\initcapa}[0]{initial capacity\xspace}
\newcommand{\initcapam}[0]{initial capacities\xspace}
\newcommand{\initcapamBig}[0]{Initial Capacities\xspace}

\newcommand{\statuslabeled}[0]{active\xspace}

\theoremstyle{plain}
\newtheorem{theorem}{Theorem}
\newtheorem{lemma}[theorem]{Lemma}
\newtheorem{definition}[theorem]{Definition}
\newtheorem{remark}{Remark}[section]

\title{Single-Source Bottleneck Path Algorithm Faster than Sorting \\ for Sparse Graphs}

\author[1]{Ran Duan \thanks{duanran@mail.tsinghua.edu.cn. Supported by a China Youth 1000-Talent grant.}}
\author[1]{Kaifeng Lyu \thanks{lkf15@mails.tsinghua.edu.cn.}}
\author[1]{Hongxun Wu \thanks{whx991201@gmail.com.}}
\author[1]{Yuanhang Xie \thanks{xieyh15@mails.tsinghua.edu.cn.}}

\affil[1]{Institute for Interdisciplinary Information Sciences, Tsinghua University}

\date{}

\begin{document}

\maketitle

\begin{abstract} \small\baselineskip=9pt
In a directed graph $G=(V,E)$ with a capacity on every edge, a \emph{bottleneck path} (or \emph{widest path}) between two vertices is a path maximizing the minimum capacity of edges in the path. For the single-source all-destination version of this problem in directed graphs, the previous best algorithm runs in $O(m+n\log n)$ ($m=|E|$ and $n=|V|$) time, by Dijkstra search with Fibonacci heap~[Fredman and Tarjan 1987]. We improve this time bound to $O(m\sqrt{\log n})$,
thus it is the first algorithm which breaks the time bound of classic Fibonacci heap when $m=o(n\sqrt{\log n})$. It is a Las-Vegas randomized approach. By contrast, the s-t bottleneck path has an algorithm with running time $O(m\beta(m,n))$~[Chechik et al. 2016], where $\beta(m,n)=\min\{k\geq 1: \log^{(k)}n\leq\frac{m}{n}\}$.
\end{abstract}

\section{Introduction}

The \emph{bottleneck path} problem is a graph optimization problem finding a path between two vertices with the maximum flow, in which the flow of a path is defined as the minimum capacity of edges on that path. The bottleneck problem can be seen as a mathematical formulation of many network routing problems, e.g. finding the route with the maximum transmission speed between two nodes in a network, and it has many other applications such as digital compositing~\cite{Fernandez1998}. It is also the important building block of other algorithms, such as the improved Ford-Fulkerson algorithm~\cite{edmonds1972theoretical,fordf1956}, and k-splittable flow algorithm~\cite{baier2002k}. The minimax path problem which finds the path that minimizes the maximum weight on it is symmetric to the bottleneck path problem, thus has the same time complexity.

\subsection{Our Results}

In a directed graph $G=(V,E)$ ($n=|V|,m=|E|$), we consider the single-source bottleneck path (SSBP) problem, which finds the bottleneck paths from a given source node $s$ to all other vertices.
In the comparison-based model, the previous best time bound for SSBP is the traditional Dijkstra's algorithm~\cite{dijkstra1959note} with Fibonacci heap~\cite{fredman1987fibonacci},
which runs in $O(m + n\log n)$ time.
Some progress has been made for slight variants of the SSBP problem: When the graph is undirected, SSBP is reducible to minimum spanning tree~\cite{Hu1961}, thus can be solved in randomized linear time \cite{karger1995randomized}; for the single-source single-destination bottleneck path ($s$-$t$ BP) problem in directed graphs, Gabow and Tarjan~\cite{gabow1988algorithms} showed that it can be solved in $O(m \log^* n)$ time, and this bound was subsequently improved by Chechik~et~al~\cite{chechik2016bottleneck} to $O(m\beta(m,n))$.
However, until now no algorithm is known to be better than Dijkstra's algorithm for SSBP in directed graphs.
And as noted in \cite{fredman1987fibonacci}, Dijkstra's algorithm can be used to sort $n$ numbers, so a ``sorting barrier'', $O(m + n \log n)$, prevents us from finding a more efficient implementation of Dijkstra's algorithm.

In this paper, we present a breakthrough algorithm for SSBP that overcomes the sorting barrier. Our main result is shown in the following theorem:
\begin{theorem}
Let $G = (V, E)$ be a directed graph with edge weights $w: E \to \R$. In comparison-based model, SSBP can be solved in expected $O\left(m \sqrt{\log n}\right)$ time.
\end{theorem}
Our algorithm is inspired by previous works on the $s$-$t$ BP problem: the $O(m \log^* n)$-time algorithm by Gabow and Tarjan~\cite{gabow1988algorithms} and the $O(m \beta(m, n))$-time algorithm by Chechik~et~al~\cite{chechik2016bottleneck}.
See Section~\ref{sec:intu} for our intuitions.

\subsection{Related Works}
A ``sorting barrier'' seemed to exist for the the Minimum Spanning Tree problem (MST) for many years \cite{Bor26,Jar30,Kru56}, but it was eventually broken by \cite{Yao75,CT76}.
Fredman and Tarjan~\cite{fredman1987fibonacci} gave an $O(m\beta(m,n))$-time algorithm by introducing Fibonacci heap. The current best time bounds for MST include randomized linear time algorithm by Karger et al~\cite{karger1995randomized}, Chazelle's $O(m\alpha(m,n))$-time deterministic algorithm~\cite{Chaz00a} and Pettie and Ramachandran's optimal approach~\cite{PR02c}.

The single-source single-destination version of the bottleneck path (s-t BP) problem is proved to be equivalent to the Bottleneck Spanning Tree (BST) problem (see~\cite{chechik2016bottleneck}). In the bottleneck spanning tree problem, we want to find a spanning tree rooted at source node $s$ minimizing the maximum edge weight in it. For undirected graph, the s-t BP can be reduced to the MST problem. For directed graph, Dijkstra's algorithm \cite{dijkstra1959note} gave an $O(n\log n+m)$-time solution using Fibonacci heap \cite{fredman1987fibonacci}. Then Gabow and Tarjan \cite{gabow1988algorithms} gave an $O(m\log^*n)$-time algorithm based on recursively splitting edges into levels. Recently, Chechik~et~al.~\cite{chechik2016bottleneck} improved the time complexity of BST and BP to randomized $O(m\beta(m,n))$ time, where $\beta(m,n)=\min\{k\geq 1: \log^{(k)}n\leq\frac{m}{n}\}$. All these algorithms are under comparison-based model. For word RAM model, an $O(m)$-time algorithm has been found by Chechik et al.~\cite{chechik2016bottleneck}.

For the all-pairs version of the bottleneck path (APBP) problem, we can sort all the edges and use Dijkstra search to obtain an $O(mn)$ time bound. For dense graphs, it has been shown that APBP can be solved in truly subcubic time. Shapira et al.~\cite{shapira2007all} gave an $O(n^{2.575})$-time APBP algorithm on vertex-weighted graphs. Then Vassilevska et al.~\cite{vassilevska2007all} showed that APBP for edge-weighted graphs can be solved in $O(n^{2+\omega/3})=O(n^{2.791})$ time based on computing $(\max, \min)$-product of real matrices, which was then improved by Duan and Pettie~\cite{duan2009fast} to $O(n^{(3+\omega)/2})=O(n^{2.686})$. Here $\omega<2.373$ is the exponent of time bound for fast matrix multiplication~\cite{Coppersmith1990, Williams:2012}.

\section{Preliminaries}

For a directed graph $G$, we denote $w(u, v)$ to be the edge weight of $(u, v) \in E$.
Without additional explanation, we use the symbol $n$ to denote the number of nodes and $m$ to denote the number of edges in $G$.

\subsection{Bottleneck Path Problems} \label{sec:def-bp}

The \textit{capacity} of a path is defined to be the minimum weight among traversed edges, i.e., if a path traverses $e_1, \dots, e_l \in E$, then the capacity of the path is $\min_{i = 1}^{l} w(e_i)$.
For any $u, v \in V$, a path from $u$ to $v$ with maximum capacity is called a \textit{bottleneck path} \textit{from} $u$ \textit{to} $v$, and we denote this maximum capacity by $b(u, v)$.

\begin{definition}
The \textit{Single-Source Bottleneck Path} (SSBP) problem is:
Given a directed graph $G = (V, E)$ with weight function $w : E \to \R$ and a source $s \in V$, output $b(s, t)$ for every $t \in V$, which is the maximum path capacity among all the paths from $s$ to $t$.
\end{definition}
We use the triple $(G, w, s)$ to denote a SSBP instance with graph $G$, weight function $w(\cdot)$ and source node $s$.

It is more convenient to present our algorithm on a slight variant of the SSBP problem. We shall call it \textit{\ModSSBPlong} (\ModSSBPshort) problem. 
We assume that the edge weight $w(e)$ of an edge $e \in E$ is either a real number or infinitely large ($+\infty$).
We say an edge $e$ is \textit{unrestricted} if $w(e) = +\infty$; otherwise we say the edge $e$ is \textit{restricted}. In the \ModSSBPshort problem, an \textit{\initcapa} $h(v)$ is given for every node $v \in V$, and the \textit{capacity} of a path is redefined to be 
the minimum between the \initcapa of the starting node and the minimum edge weights in the path, i.e., if the path starts with the node $v \in V$ and traverses $e_1, \dots, e_l$, then its capacity is $\min\left(\{h(v)\} \cup \{w(e_i)\}_{i = 1}^{l}\right)$. 
For any $v \in V$, a path ended with $v$ with maximum capacity is called a \textit{bottleneck path ended with} $v$, and we denote this maximum capacity as $d(v)$.

\begin{definition}
The \textit{\ModSSBPlong} (\ModSSBPshort) problem is:
Given a directed graph $G = (V, E)$ with weight function $w : E \to \R \cup \{+\infty\}$ and \initcapa function $h: V \to \R \cup \{\pm \infty\}$, output $d(v)$ for every $v \in V$, which is the maximum path capacity among all the paths ended with $v$. 
\end{definition}
We use the triple $(G, w, h)$ to denote an \ModSSBPshort instance with graph $G$, weight function $w(\cdot)$, and inital capacity function $h(\cdot)$.

Note that \ModSSBPshort and SSBP are equivalent under linear-time reductions. Given an \ModSSBPshort instance, we construct a new graph $G'$ from $G$ by adding a new node $s$ which has outgoing edges with weight $h(v)$ to all the nodes $v \in V$ having $h(v) > -\infty$, then it suffices to find bottleneck paths from $s$ to all other nodes in $G'$.
On the other hand, a SSBP instance $(G, w, s)$ can be easily reduced to the \ModSSBPshort instance $(G, w, h)$, where $h(s) = \max_{e\in E}\{w(e)\}$ and $h(v) = -\infty$ for all $v \neq s$.

\subsection{Dijkstra's Algorithm for SSBP and \ModSSBPshort} \label{sec:dijk}

SSBP can be easily solved using a variant of Dijkstra's algorithm~\cite{dijkstra1959note}. In this algorithm, each node is in one of the three status: \textit{unsearched}, \textit{\statuslabeled}, or \textit{scanned}. We associate each node $v \in V$ with a \textit{label} $d'(v)$, which is the maximum path capacity among all the paths from $s$ to $v$ that only traverse scanned nodes or $v$.

Initially, all the nodes are unsearched except $s$ is \statuslabeled, and we set $d'(s) = +\infty$ and $d'(v) = -\infty$ for all $v \neq s$. We repeat the following step, which we call the \textit{Dijkstra step}, until none of nodes is \statuslabeled:
\begin{itemize}
    \item Select an \statuslabeled node $u$ with maximum label and mark $u$ as scanned. For every outgoing edge $(u, v)$ of $u$, update the label of $v$ by
\begin{equation} \label{eq:dijkupd}
    d'(v) \gets \max\{d'(v), \min\{d'(u), w(u, v)\}\},
\end{equation}
and mark $v$ as \statuslabeled if $v$ is unsearched.
\end{itemize}
We use priority queue to maintain the order of labels for \statuslabeled nodes. This algorithm runs in $O(m + n \log n)$ time when Fibonacci heap \cite{fredman1987fibonacci} is used.

The algorithm we introduced above can also be adapted for solving \ModSSBPshort. The only thing we need to change is that in the initial stage all nodes are \statuslabeled and $d'(v) = h(v)$ for every $v \in V$. The resulting algorithm again runs in $O(m + n \log n)$ time. We shall call these two algorithms as \textit{Dijkstra's algorithm for SSBP} and \textit{Dijkstra's algorithm for \ModSSBPshort}, or simply call any of them \textit{Dijkstra's algorithm} when no confusion can arise.

\subsection{Weak and Strong Connectivity in Directed Graph}

We also need some definitions about connectivity in graph theory in this paper. A directed graph is said to be \textit{weakly-connected} if it turns to be a connected undirected graph when changing all of its directed edges to undirected edges. A directed graph is said to be \textit{strongly-connected} if every pair of nodes can be reached from each other. A \textit{weakly-} (or \textit{strongly-}) \textit{connected components} is defined to be a maximal weakly- (or strongly-) connected subgraph.

\section{Intuitions for SSBP} \label{sec:intu}

If all the edge weights are integers and are restricted in $\{1, \dots, c\}$, then SSBP can be solved in $O(m + c)$ time using Dijkstra's algorithm with bucket queue. If the edge weights are not necessarily small integers but all the edges given to us are already sorted by weights, then we can replace the edge weights by their ranks and use Dijkstra's algorithm with bucket queue to solve the problem in $O(m)$ time. However, edges are not sorted in general. If we sort the edges directly, then a cost of $\Omega(m \log m)$ time is unavoidable in a comparison-based model, which is more expensive than the $O(m + n \log n)$ running time of Dijkstra's algorithm.

Our algorithm is inspired by previous works on the single-source single-destination bottleneck path problem ($s$-$t$ BP): the $O(m \log^* n)$-time algorithm by Gabow and Tarjan~\cite{gabow1988algorithms} and the $O(m \beta(m, n))$-time algorithm by Chechik~et~al~\cite{chechik2016bottleneck}.
Gabow and Tarjan's algorithm for $s$-$t$ BP consists of several stages. Let $b(s, t)$ be the capacity of a bottleneck path from $s$ to $t$. Initially, we know that $b(s, t)$ is in the interval $(-\infty, +\infty)$. In each stage, we narrow the interval of possible values of $b(s, t)$.
Assume that $b(s, t)$ is known to be in the range $(l, r)$. Let $m_{(l,r)}$ be the the number of edges with weights in the range $(l, r)$ and $k$ be a parameter.
By applying the median-finding algorithm \cite{blum1973time} repeatedly, we choose $k$ thresholds $\lambda_1, \cdots, \lambda_k$ to split $(l, r)$ into $k + 1$ subintervals $(l, \lambda_1), [\lambda_1, \lambda_2), [\lambda_2, \lambda_3), \dots, [\lambda_{k-1}, \lambda_k), [\lambda_k, r)$ such that for each subinterval, there are $O(m_{(l,r)} / k)$ edges of weight in it.
Gabow and Tarjan then show that \textit{locating} which subinterval contains $b(s, t)$ can be done in $O(m_{(l,r)})$ time by incremental search. Finally, the $O(m \log^* n)$ running time bound is achieved by setting $k$ appropriately at each stage.

The algorithm by Chechik~et~al.~is based on the framework of Gabow and Tarjan's algorithm, but instead of selecting the thresholds $\lambda_1, \dots, \lambda_k$ by median-finding repeatedly in $O(m_{(l, r)} \log k)$ time, in this algorithm we select the $k$ thresholds by randomly sampling in edge weights, and sort them in $O(k \log k)$ time. These thresholds partition the edges evenly with high probability, but it requires $\Omega(m \log k)$ time to compute the partition explicitly.
Then they show that actually we can locate which subinterval contains $b(s, t)$ in $O(m_{(l, r)} + nk)$ (or $O(m_{(l, r)} + n \log k)$) time, without computing the explicit partition. The time bound for the overall algorithm is again obtained by setting $k$ appropriately at each stage.

We adapt Chechik~et~al.'s framework for the $s$-$t$ BP problem to the SSBP problem.
Our SSBP algorithm actually works on an equvialent problem called \ModSSBPshort. In \ModSSBPshort, there is no fixed source but every node has an \initcapa, and for all destination $t \in V$ we need to compute the capacity $d(t)$ of a bottleneck path ended with $t$ (See Section \ref{sec:def-bp} for details). Instead of locating the subinterval for a single destination $t$, our algorithm locates the subintervals for all destinations $t \in V$. Thus we adopt a divide-and-conquer methodology. At each recursion, we follow the idea from Chechik~et~al.~\cite{chechik2016bottleneck} to randomly sample $k$ thresholds. Then we \textit{split} the nodes into $k + 1$ levels $V_0, \dots, V_k$, where the $i$-th level contains the nodes $t$ that have $d(t)$ in the $i$-th subinterval ($0 \le i \le k$). For each level $V_i$ of nodes, we compute $d(t)$ for every $t \in V_i$ by reducing to solve the SSBP on a subgraph consisting of all the nodes in $V_i$ and some of the edges connecting them. We set $k$ to be fixed in all recursive calls, and the maximum recursion depth is $O(\log n / \log k)$ with high probability.

The split algorithm becomes the key part of our algorithm.
Note that at each recursion, we should reduce or avoid the use of operations that cost time $O(\log k)$ per node or per edge (e.g., binary searching for the subinterval containing a particular edge weight). This is because that, for example, if we execute an $O(\log k)$-time operation for each edge at each recursion, then the overall time cost is $O(m \log k) \cdot O(\log n / \log k) = O(m \log n)$, which means no improvement comparing with the previous $O(m + n \log n)$-time Dijkstra's algorithm.
Surprisingly, we can design an algorithm such that whenever we execute an $O(\log k)$-time operation, we can always find one edge that does not appear in any subsequent recursive calls. Thus total time complexity for such operations is $O(m \log k)$, which gives us some room to obtain a better time complexity by adjusting the parameter~$k$.

\section{Our Algorithm}
Our algorithm for SSBP is as follows: Given a SSBP instance $(G, w, s)$, we first reduce the SSBP problem to an \ModSSBPshort instance $(G, w, h)$, and then use a recursive algorithm (Figure \ref{fig:algomain}) to solve the \ModSSBPshort problem. The reduction is done by setting $h(s) = \max_{e\in E}\{w(e)\}$ and $h(v) = -\infty$ for all $v \neq s$ as described in the preliminaries.

For convenience, we assume that in the original graph, all  edge weights are distinct.
This assumption can be easily removed.

A high-level description of our recursive algorithm for \ModSSBPshort is shown in Figure \ref{fig:algomain}.
For two set $A$ and $B$, $A \uplus B$ stands for the union of $A$ and $B$ with the assumption that $A \cap B = \varnothing$. We use $E^{(r)}$ to denote the set of restricted edges in $G$, and similarly we use $E^{(r)}_i$ to denote the set of restricted edges in $G_i$ for each \ModSSBPshort instance $(G_i, w_i, h_i)$. When the thresholds $\lambda_0, \dots, \lambda_{k+1}$ are presence, we define the \textit{index} of $x$ for every $x \in \R$ to be the unique index $i$ such that $\lambda_i \le x < \lambda_{i+1}$, and we denote it as $I(x)$. For $x = \pm \infty$, we define $I(-\infty) = 0, I(+\infty) = l$. Note that all the subgraphs $G_i$ at Line \ref{algo:line:red} are disjoint. We denote $r = \lvert E \rvert - \sum_{i=0}^{l} \lvert E_i \rvert$ to be the total number of edges in $E$ that do not appear in any recursive calls of $(G_i, w_i, h_i)$. For an edge $(u, v) \in E$, if $u$ and $v$ belong to different levels, then we say that $(u, v)$ is \textit{cross-level}. If $u$ and $v$ belong to the same level $V_i$ and $w(u, v) < \lambda_i$, then we say that $(u, v)$ is \textit{below-level}; conversely, if $w(u, v) \ge \lambda_{i+1}$ then we say that $(u, v)$ is \textit{above-level}.

\begin{figure}[htbp]
\begin{algorithmic}[1]
\Require Directed graph $G=(V,E)$, weight function $w(\cdot)$, \initcapa function $h(\cdot)$; Parameter $k \ge 1$.
\Ensure For each $v \in V$, output $d(v)$, the capacity of a bottleneck path ended with $v$.
\If {$G$ is not weakly-connected} \label{algo:line:weak1}
    \State Compute $\{d(v)\}_{v \in V}$ in each weakly-connected component recursively. \label{algo:line:weak2}
\EndIf
\State Let $E^{(r)} = \{ e \in E \mid w(e) < +\infty\}$ be the set of restricted edges.
\If {$\lvert E^{(r)} \rvert \le 1$}
    \State Compute $\{d(v)\}_{v \in V}$ in linear time and exit. \label{algo:line:restr1} \Comment{Section \ref{sec:restr1}}
\EndIf
\State Sample $l = \min\{k, \lvert E^{(r)} \rvert\}$ distinct edges from $E^{(r)}$ uniformly randomly.
\State Sort the sampled edges by weights, and let $\lambda_1 < \cdots < \lambda_l$ be their weights. \label{algo:line:sort}
\State Let $\lambda_0 = -\infty, \lambda_{l + 1} = +\infty$.
\State Split $V$ into $l+1$ levels: $V_0 \uplus V_1 \uplus \cdots \uplus V_l$, where $V_i = \{v \in V \mid \lambda_i \le d(v) < \lambda_{i+1} \}$ \label{algo:line:spl}.

\Comment{Section \ref{sec:spl}} 
\State For every level $V_i$, reduce the computation of $\{d(v)\}_{v \in V_i}$ to a new \ModSSBPshort instance $(G_i, w_i, h_i)$, where $G_i = (V_i, E_i)$ is a subgraph of $G$ consisting of all the nodes in $V_i$ and some of edges that connect them. Solve each $(G_i, w_i, h_i)$ instance recursively. \label{algo:line:red}
\end{algorithmic}
\caption{Main Algorithm}  \label{fig:algomain}
\end{figure}

Besides the problem instance $(G, w, h)$ of \ModSSBPshort, our algorithm requires an additional integral parameter $k$. We set the parameter $k = 2^{\Theta\left(\sqrt{\log n}\right)}$ throughout our algorithm. The value of the parameter $k$ does not affect the correctness of our algorithm, but it controls the number of recursive calls at each recursion.

At each recursion, our algorithm first checks if $G$ contains only one weakly-connected component. If not, then our algorithm calls itself to compute $d$ in each weakly-connected component recursively. Now we can assume that $G$ is weakly-connected (so $n \le m$).

If the number of restricted edges is no more than $1$, we claim that we can compute $d(v)$ for all~$v$ in linear time. The specific algorithm will be introduced in Section \ref{sec:restr1}.

\begin{lemma} \label{lam:restr1}
\ModSSBPshort can be solved in $O(m)$ time if there is at most one restricted edge.
\end{lemma}

If the number of restricted edges is more than $1$, then our algorithm first sample $l = \min\{k, \lvert E^{(r)} \rvert\}$ distinct edges from $E^{(r)}$ uniformly randomly and sort them by weights, that is, if the number of restricted edges is more than $k$, then we sample $k$ distinct restricted edges and sort them; otherwise we just sort all the restricted edges. Let $\lambda_i$ be the weight of the edge with rank $i$ ($1 \le i \le l$) and $\lambda_0 = -\infty$, $\lambda_{l + 1} = +\infty$.

Next, we split $V$ into $l + 1$ levels of nodes $V = V_0 \uplus V_1 \uplus \cdots V_l$, where the $i$-th level of nodes is $V_i = \{v \in V \mid I(d(v)) = i \} = \{v \in V \mid \lambda_i \le d(v) < \lambda_{i + 1}\}$. The basic idea of the split algorithm is: we run Dijkstra's algorithm for \ModSSBPshort on the graph produced by mapping every edge weight $w(e)$ and initial capacity $h(v)$ in $G$ to their indices $I(w(e))$ and $I(h(v))$, and we obtain the final label value $d'(v)$ for each node $v \in V$ (Remember that $d'(v)$ is the label of $v$ in Dijkstra's algorithm). It is easy to show that the final label value $d'(v)$ equals $I(d(v))$, so the nodes can be easily split into levels according to their final labels.
The specific split algorithm will be introduced in Section \ref{sec:spl}. The time complexity for a single splitting is given below.
In Theorem \ref{thm:overall-k} we show that this implies that the total time cost for splitting is $O(m \log n / \log k + m \log k)$.

\begin{lemma} \label{lam:spl}
Splitting $V$ into levels at Line \ref{algo:line:spl} can be done in $O(m + r \log k)$ (Recall that $r$ is the number of edges that do not appear in any subsequent recursive calls).
\end{lemma}

Finally, for every level $V_i$, we compute $d(\cdot)$ for nodes in this level by reducing to a new \ModSSBPshort instance $(G_i, w_i, h_i)$, where $G_i$ is a subgraph of $G$ consisting of all the nodes in $V_i$ and some of edges that connect them. We solve each new instance by a recursive call. The construction of $(G_i, w_i, h_i)$ is as follows:
\begin{itemize}
    \item $G_i = (V_i, E_i)$, where $V_i$ is the nodes at level $i$ in $G$, and $E_i$ is the set of edges 
    which connect two nodes at level $i$ and are not below-level, i.e., $E_i = \{ (u, v) \in E \mid u, v \in V_i, w(u, v) \ge \lambda_i \}$;
    \item For any $e \in E_i$, $w_i(e) = +\infty$ if $e$ is above-level; otherwise $w_i(e) = w(e)$;
    \item For any $v \in V_i$, $h_i(v) = \max\left( \{h(v)\} \cup \left\{ w(u, v) \in E \mid u \in V_{i+1} \uplus \cdots \uplus V_l \right\} \right)$.
\end{itemize}
\begin{lemma} \label{lam:red}
We can construct all the new \ModSSBPshort instances $(G_0, w_0, h_0), \dots, (G_l, w_l, h_l)$ in $O(m)$ time.
For $v \in V_i$, the value of $d(v)$ in the instance $(G_i, w_i, h_i)$ exactly equals to the value of $d(v)$ in the instance $(G, w, h)$.
\end{lemma}
\begin{proof}
We can construct all these instances $(G_i, w_i, h_i)$ for all $0 \le i \le l$ by linearly scanning all the nodes and edges, which runs in $O(m)$ time. We prove the correctness by transforming $(G, w, h)$ to $(G_i, w_i, h_i)$ step by step, while preserving the values of $d(v)$ for all $v \in V_i$.

By definition, $\lambda_i \le d(v) < \lambda_{i+1}$ for all $v \in V_i$. We can delete all the nodes at level less than $i$ and delete all the edges with weight less than $\lambda_i$, since no bottleneck path ended with a node in $V_i$ can traverse them. Also, for every edge $e$ with weight $w(e) \ge \lambda_{i+1}$, we can replace the edge weight with $+\infty$ since $w(e)$ is certainly not the minimum in any path ended with a node in $V_i$. 

For every edge $e = (u, v)$ where $u \in V_{i+1} \uplus \cdots \uplus V_l$ and $v \in V_i$, the edge weight $w(e)$ must be less than $\lambda_{i+1}$, otherwise $d(v) \ge \min\{d(u), w(u, v)\} \ge \lambda_{i+1}$ leads to a contradiction. Thus contracting all the nodes in $V_{i+1} \uplus \cdots \uplus V_l$ to a single node $v_0$ with infinite initial capacity is a transformation preserving the values of $d(v)$ for all $v \in V_i$. Finally, our construction follows by taking $h_i(v)$ to be the maximum between the weight of incoming edges from $v_0$ and the initial capacity $h(v)$ for every $v \in V_i$.
\end{proof}

\begin{remark}
In any subsequent recursive calls of $(G_i, w_i, h_i)$, neither cross-level nor below-level edges will appear, and all the above-level edges will become unrestricted. Also, it is easy to see that $r$ is just the total number of cross-level and below-level edges (Recall that $r$ is the number of edges that do not appear in any subsequent recursive calls).
\end{remark}

\subsection{Running Time}
First we analyze the maximum recursion depth. The following lemma shows that randomly sampled thresholds evenly partition the restricted edges with high probability.
\begin{lemma} \label{lam:samplee}
Let $E^{(r)} = \left\{e_1, \dots, e_q\right\}$ be $q$ restricted edges sorted by their weights in $E$. Let $f_1, \dots, f_k$ be $k \ge 2$ random edges sampled from $E^{(r)}$ such that $w(f_1) < w(f_2) < \cdots < w(f_k)$. Let $\lambda_i = w(f_i)$ for $i = 1, \dots, k$, and $\lambda_0 = -\infty$, $\lambda_{k + 1} = +\infty$. Let $F_i = \{ e \in E \mid \lambda_i \le w(e) < \lambda_{i+1} \}$. Then for every $t > 0$, $\max_{0 \le i \le k} \{\lvert F_i \rvert\} < t q \log k / k$ holds with probability $1 - k^{-\Omega(t)}$.
\end{lemma}
\begin{proof}
Let $M = tq\log k / k$. If $\max_{0 \le i \le k} \{\lvert F_i \rvert\} \ge M$, then there exists an edge $e_p$ such that $e_p$ is chosen but for any $p + 1 \le j < p + M$, $e_j$ is not chosen. Note that when $p$ is given, this event happens with probability $\le k \cdot (1/q) \cdot \prod_{i=1}^{k-1} ((q-M-i)/(q-i)) \le (k/q) \cdot (1 - M/q)^{k-1}$. By the union bound for all possible $p$, we have
\begin{equation*}
\Pr\left[\max_{0 \le i \le k} \{\lvert F_i \rvert\} \ge t q \log k/k\right] \le k (1 - t \log k / k)^{k-1} \le k^{-\Omega(t)},
\end{equation*}
which completes the proof.
\end{proof}
For our purposes it is enough to analyze the case that $k = 2^{\Omega\left(\sqrt{\log n}\right)}$. The following lemma gives a bound for the maximum recursion depth using Lemma \ref{lam:samplee}.
\begin{lemma}
For $k = 2^{\Omega\left(\sqrt{\log n}\right)}$ where $n$ is the number of nodes at the top level of recursion, the maximum recursion depth is $O(\log n / \log k)$ with probability $1 - n^{-\omega(1)}$.
\end{lemma}
\begin{proof}
It is not hard to see that the total number of recursive calls of our main algorithm is $O(m)$.
Applying Lemma \ref{lam:samplee} with $t = \Theta(\log n)$ and the union bound for all recursive calls, we know that with probability at least $1 - k^{-\Omega(t)} \cdot O(m) = 1 - n^{-\Omega(\log n)^{3/2}}$, after every split with $\lvert E^{(r)} \rvert > k$,
the number of restricted edges in $G_i$ is less than $(t \log k / k) \cdot \lvert E^{(r)} \rvert$ for every $(G_i, w_i, h_i)$. Thus after $O(\log m / \log (k/(t \log k))) = O(\log n / \log k)$ levels of recursion, every \ModSSBPshort instance $(G, w, h)$ has $\lvert E^{(r)}\rvert \le k$, and this means that in any recursive call of $(G_i, w_i, h_i)$, the graph $G_i$ has at most one restricted edge, which will be directly solved at Line~\ref{algo:line:restr1}.
\end{proof}

The overall time complexity of our algorithm is given by the following theorem:
\begin{theorem} \label{thm:overall-k}
For $k = 2^{\Omega\left(\sqrt{\log n}\right)}$, with probability $1 - n^{-\omega(1)}$, our main algorithm shown in Figure \ref{fig:algomain} runs in $O(m \log n / \log k + m \log k)$ time.
\end{theorem}
\begin{proof}
Let $r = \lvert E \rvert - \sum_{i=0}^{l} \lvert E_i \rvert$, $r' = \lvert E^{(r)} \rvert - \sum_{i=0}^{l} \lvert E_i^{(r)} \rvert$. First we show that the running time in each recursive call is
$O(m + (r + r') \log k)$.

In each recursive call of our algorithm, the time cost for sorting at Line \ref{algo:line:sort} is $O(l \log l)$. For the sample edge $e_i$ with rank $i$, either $V_i$ is not empty, or this edge is cross-level, below-level, or above-level.  Let $l_1$ be the number of edges in the former case, and $l_2$ be the number of edges in the latter case. For the former case, note that we only run the split algorithm for weakly-connected graphs, so there are at least $l_1 - 1$ cross-level edges, which implies $l_1 \le r + 1$. For the latter case, $e_i$ becomes unrestricted or does not appear for every $G_i$, so $l_2 \le r'$. Thus $l = l_1 + l_2 \le r + r' + 1$ and the time cost for sorting is $O((r + r') \log k)$.

By Lemma \ref{lam:spl}, the split algorithm runs in $O(m  + r \log k)$ time in each recursive call. All other parts of our algorithm run in linear time. Thus the running time for each recursion is $O(m + (r + r') \log k)$.
Note that the recursion depth is $O(\log n / \log k)$ with probability $1 - n^{\omega(1)}$. We can complete the proof by adding the time cost for all the $O(m)$ recursive calls together.
\end{proof}

Finally, we can get our main result by setting $k = 2^{\Theta\left(\sqrt{\log n}\right)}$.
\begin{theorem}
SSBP can be solved in $O(m \sqrt{\log n})$ time with high probability.
\end{theorem}
\begin{remark}
The above time bound is also true for expected running time. It can be easily derived from the fact that the worst-case running time is at most $n^{O(1)}$.
\end{remark}

In the rest of this section, we introduce the algorithm for \ModSSBPshort with at most one restricted edge and the split algorithm.

\subsection{Algorithm for the Graph with at most One Restricted Edge} \label{sec:restr1}

We introduce our algorithm for the graph with at most one restricted edge in the following two lemmas.

\begin{lemma} \label{thm:zero-res}
For a given \ModSSBPshort instance $(G, w, h)$, if there is no restricted edge in $G$, then the values of $d(v)$ for all $v \in V$ can be computed in linear time.
\end{lemma}
\begin{proof}
$G$ contains only unrestricted edges, so for every $v \in V$, $d(v)$ is just equal to the maximum value of $h(u)$ among all the nodes $u$ that can reach $v$. If $v_1$ and $v_2$ are in the same strongly-connected component, then $d(v_1) = d(v_2)$. Thus we can use Tarjan's algorithm~\cite{tarjan1972depth} to contract every strongly-connected component to a node. The initial capacity of a node is the maximum $h(u)$ for all $u$ in the component, and the capacity of an edge between nodes is the maximum among edges connecting components. Then Dijkstra approach on DAG takes linear time.
\end{proof}

\begin{lemma}
For a given \ModSSBPshort instance $(G, w, h)$, if there is exactly one restricted edge in $G$, then the values of $d(v)$ for all $v \in V$ can be computed in linear time.
\end{lemma}
\begin{proof}
Let $e_0 = (u_0, v_0)$ be the only restricted edge in $G$.
There are two kinds of paths in $G$:
\begin{enumerate}
    \item Paths that do not traverse $e$. We remove $e$ from $G$ and use the algorithm in Lemma \ref{thm:zero-res} to get $d(v)$ for every node $v \in V$.
    \item Paths that traverse $e$. Note that $d(u_0)$ got in the previous step is the maximum capacity to $u_0$ through only unrestricted edges. Then we update $d(v)$ by $\max\{d(v), \min\{d(u_0), w(u_0, v_0)\}\}$ for every node $v$ that can be reached from $v_0$.
\end{enumerate}
We output the values of $d(\cdot)$ after these two kinds of updates, then all the paths should have been taken into account.
\end{proof}

\subsection{Split} \label{sec:spl}

Now we introduce the split algorithm at Line \ref{algo:line:spl} in our main algorithm. As before, we use the notation $I(x)$ for the index of a value $x$ and $d'(v)$ is the label of $v$ in Dijkstra's algorithm.
The goal of this procedure is to split $V$ into $l + 1$ levels, $V = V_0 \uplus V_1 \uplus \cdots \uplus V_l$, where $V_i = \{v \in V \mid I(d(v)) = i\}$.
We need to show that this can be done in $O(m + r \log k)$ time, where $r = \lvert E \rvert - \sum_{i=0}^{l} \lvert E_i \rvert$ is the total number of edges in $E$ that do not appear in any $(G_i, w_i, h_i)$.

A straightforward approach to achieve this goal is to use Dijkstra's algorithm as described in Section \ref{sec:dijk}.
We map all the edge weights and \initcapam to their indices using binary searches, and run Dijkstra's algorithm for \ModSSBPshort. The output should be exactly $d'(v) = I(d(v))$ for every $v$. However, this approach is rather inefficient. Evaluating the index $I(x)$ for a given $x$ requires $\Omega(\log l)$ time in a comparison-based model, thus in total, this algorithm needs $\Omega(n \log l)$ time to compute indices, and this does not meet our requirement for the running time.

The major bottleneck of the above algorithm is the inefficient index evaluations. Our algorithm overcomes this bottleneck by reducing the number of index evaluations for both edge weights and \initcapam to be at most $O(r + n / \log l)$.

\subsubsection{Index Evaluation for Edge Weights} \label{sec:indedge}
First we introduce our idea to reduce the number of index evaluations for edge weights. Recall that in Dijkstra's algorithm for \ModSSBPshort we maintain a label $d'(v)$ for every $v \in V$. In every Dijkstra step, we extract an \statuslabeled node $u$ with the maximum label, and for all edges $(u, v) \in E$ we compute $\min\{d'(u), I(w(u, v))\}$ to update $d'(v)$. In the straightforward approach we evaluate every $I(w(u, v))$ using binary search, but actually this is a big waste:
\begin{enumerate}
    \item If $w(u, v) \ge \lambda_{d'(u)}$, then $\min\{d'(u), I(w(u, v))\} = d'(u)$, so there is no need to evaluate $I(w(u, v))$.
    \item If $w(u, v) < \lambda_{d'(u)}$, then $\min\{d'(u), I(w(u, v))\} = I(w(u, v))$, so we do need to evaluate $I(w(u, v))$. However, it can be shown that $(u, v)$ is either a cross-level edge or a below-level edge, so $(u, v)$ will not appear in any subsequent recursive calls of $(G_i, w_i, h_i)$.
\end{enumerate}
Using the method discussed above, we can reduce the number of index evaluations for edge weights to be at most $r = \lvert E \rvert - \sum_{i=0}^{l} \lvert E_i \rvert$ in Dijkstra's algorithm. Lemma \ref{lam:qe} gives a formal proof for this.

\subsubsection{Index Evaluation for \initcapamBig}

Now we introduce our idea to reduce the number of index evaluations for \initcapam. Recall that in Dijkstra's algorithm, we need to initialize the label $d'(v)$ to be $I(h(v))$ for each $v \in V$, and maintain a priority queue for the labels of all \statuslabeled nodes. If we evaluate every $I(h(v))$ directly, we have to pay a time cost of $\Omega(n \log l)$.

In our split algorithm, we first find a spanning tree $T$ of $G$ after replacing all the edges with undirected edges. Then we partition the tree $T$ into $b = O(n/s)$ edge-disjoint (but not necessarily node-disjoint) subtrees, $T_1, \dots, T_b$, each of size $O(s)$. Theorem~\ref{thm:partition} shows that this partition can be found in $O(n)$ time, and the proof is given in Appendix~\ref{app:partition}.

\begin{theorem} \label{thm:partition}
Given a tree $T = (V, E)$ with $n$ nodes and given an integer $1 \le s \le n$, there exists a linear time algorithm that can partition $T$ into edge-disjoint subtrees, $T_1, \dots, T_b$, such that the number of nodes in each subtree is in the range $[s, 3s)$.
\end{theorem}

We form $b$ groups of nodes, $U_1, \dots, U_b$, where $U_i$ is the group of nodes that are in the $i$-th subtree $T_i$.
In the running of Dijkstra's algorithm, we divide the \statuslabeled nodes in $V$ into two kinds:
\begin{enumerate}
    \item \textbf{Updated node}. This is the kind of node $v$ that has $d'(v)$ already been updated by Dijkstra's update rule \eqref{eq:dijkupd}, which means $d'(v) = \min\{d'(u), I(w(u, v))\} \geq I(h(v))$ for some $u$ after a previous update. The value of $\min\{d'(u), I(w(u, v))\}$ is evaluated according to Section \ref{sec:indedge}, so the value of $d'(v)$ can be easily known. We can store such nodes in buckets to maintain the order of their labels.
    \item \textbf{Initializing node}. This is the kind of node $v$ whose $d'(v)$ has not been updated by Dijkstra's update rule \eqref{eq:dijkupd}, so $d'(v) = I(h(v))$. However, we do not store the value of $I(h(v))$ explicitly.
    For each group $U_i$, we maintain the set of initializing nodes in $U_i$. We only compute the value of $I(h(v))$ when $v$ is the maximum in its group, and use buckets to maintain the maximum values from all groups. If each group has size $s=O(\log k)$, the maximum can be found in $O(\log k)$ time by brute-force search.
\end{enumerate}

At each iteration in Dijkstra's algorithm, we extract the active node with the maximum label among the updated nodes and initializing nodes and mark it as scanned. For the case that the maximum node $v \in U_j$ is an initializing node, we remove $v$ from $U_j$, and compute $d'(u) = I(h(u))$ for the new maximum node $u$.
However, if we compute this value directly using an index evaluation for $h(u)$, then we will suffer a total cost $\Omega(n \log l)$, which is rather inefficient.

The idea to speed up is to check whether $d'(u) = d'(v)$ before performing an index evaluation. This can be done in $O(1)$ time since we can know the corresponding interval of $h(u)$ from the value of $d'(v)$ if $I(h(u)) = d'(v)$. We only actually evaluate $I(h(u))$ if $d'(u) \neq d'(v)$ after the Dijkstra step scans all nodes with level $d'(v)$. In this way, we can always ensure that the number of index evaluations in a group never exceeds the number of different final label values $d'(\cdot)$ in this group. Indeed, it can be shown that if there are $c$ different final label values $d'(\cdot)$ in a group $U_i$, then there must be at least $c - 1$ cross-level edges in $T_i$, which implies that the number of index evaluations for \initcapam should be no greater than $r + b$. (Remember $b$ is the number of groups in the partition.) Lemma \ref{lam:diffdp} gives a formal proof for this.

\subsubsection{The $O(m + r \log k )$-time Split Algorithm}

\begin{figure}[htb]
\begin{algorithmic}[1]
\Require Directed graph $G=(V,E)$, weight function $w(\cdot)$, \initcapa function $h(\cdot)$ , $l + 2$ thresholds $-\infty = \lambda_0 < \lambda_1 < \cdots < \lambda_{l} < \lambda_{l + 1} = +\infty$
\Ensure Split $V$ into $V_0 \uplus V_1 \uplus \cdots \uplus V_l$ where $V_i = \{v \in V \mid \lambda_i \le d(v) < \lambda_{i+1} \}$.
\State Let $s = \min\{\lceil \log l \rceil, n\}$. Find a spanning tree $T$ of $G$ after replacing all the edges with undirected edges and partition $T$ into $b = O(n / s)$ edge-disjoint subtrees $T_1, \dots, T_b$, each of size $O(s)$.
\State Form $b$ groups of nodes, $U_1, \dots, U_b$, where $U_i$ is the group of nodes that are in the $i$-th subtree~$T_i$.
\State For each group $U_i$, find the maximum node $v$ and evaluate $I(h(v))$. Initialize $l + 1$ buckets $B_0, \dots, B_l$ which store the groups according to the index of $h(\cdot)$ of their maximum nodes.
\State Initialize $l+1$ buckets $C_0, \dots, C_l \gets \varnothing$ which store the \statuslabeled nodes according to $d'(\cdot)$.
\For{$i \gets l, l-1, \dots, 1, 0$}
    \ForAll{$U\in B_i$}
        \ForAll{$u \in U$ with $h(u) \ge \lambda_i$}
            \State Delete $u$ from $U$ and put $u$ into $C_i$
            \State $d'(u)\gets i$
        \EndFor
    \EndFor
    \While{$C_i \neq \varnothing$}
        \State Extract a node $u$ from $C_i$
        \ForAll{$(u, v) \in E$}
            \State $\bar{w} \gets \min\{d'(u), I(w(u, v))\}$. (Evaluate $I(w(u, v))$ only if $w(u, v) < \lambda_{d'(u)}$)
            \If{$\lambda_{\bar{w}} > h(v)$}
                \If{$d'(v)$ does not exist \textbf{or} $\bar{w}>d'(v)$}
                    \State Delete $v$ from $C_{d'(v)}$
                    \State Delete $v$ from the group $U$ containing it (if any) and put $v$ into $C_{\bar{w}}$
                    \State $d'(v)\gets \bar{w}$
                \EndIf
            \EndIf
        \EndFor
    \EndWhile
    \While{$B_i$ contains at least one non-empty group}
        \State Extract an non-empty group $U$ from $B_i$
        \State Find the maximum node $u$ in $U$, evaluate $I(h(u))$, and put $U$ into $B_{I(h(u))}$
    \EndWhile
\EndFor
\State Split $V$ according to $d'$ and return
\end{algorithmic}
\caption{Split Algorithm} \label{fig:algosplit}
\end{figure}

Now we are ready to introduce our split algorithm in full details. A pseudocode is shown in Figure~\ref{fig:algosplit}. During the search at Line 5 - 21, $B_i$ may contain groups with maximum nodes not at the $i$-th level, e.g., when the maximum node in a group $U$ is deleted at Line 17 and $I(h(u))$ of the new maximum node $u$ in $U$ has not been evaluated yet. We have the following observations:
\begin{itemize}
    \item For all $v \in V$, $d'(v)$ is non-decreasing, and at the end we have $d'(v) = I(d(v))$.
    \item Only at Line 3, 13 and 21 we need to evaluate the index of $h(\cdot)$ of a node or the index of an edge. Each index evaluation costs $O(\log l)$ time.
    \item The numbers of executions of the while loops at Line 10 - 18, Line 19 - 21 are bounded by $O(n)$.
    \item The number of times entering the for loop at Line 7 - 9 can be bounded by the number of index evaluations at Line 3 and 21. Each loop costs $O(\log l)$ time.
\end{itemize}

Our algorithm is an implementation of Dijkstra's algorithm, so the correctness is obvious. The running time analysis is based on the following lemmas:
\begin{lemma} \label{lam:qe}
If we evaluate $I(w(u,v))$ at Line 13, then the edge $(u,v)$ will not be in any recursive calls of $(G_i, w_i, h_i)$.
\end{lemma}
\begin{proof}
We evaluate $I(w(u,v))$ only if $w(u, v) < \lambda_{d'(u)}$, so $\bar{w} = I(w(u, v))$ right after Line 13. Since $u$ has already been scanned, $d'(u)$ here is just its final value $I(d(u))$. Note that $I(d(v)) \ge \min\{I(d(u)), I(w(u, v))\} = I(w(u, v))$. 
If finally $I(d(v)) > I(w(u,v))$, then $I(w(u,v))$ is smaller than both $I(d(u))$ and $I(d(v))$, thus $(u,v)$ is a below-level edge or cross-level edge. If $I(d(v))=I(w(u,v))$, then $I(d(v)) < I(d(u))$ and $(u, v)$ is a cross-level edge.
\end{proof}

\begin{lemma} \label{lam:diffdp}
At Line 3 and 21, if we evaluate $I(h(u))$ for $c$ nodes $u$ in some group $U_i$, then the number of different final label values $d'(\cdot)$ in $U_i$ is at least $c$. Thus, the number of edges in the subtree $T_i$ corresponding to $U_i$ that do not appear in any recursive calls of $(G_i, w_i, h_i)$ is at least $c - 1$.
\end{lemma}
\begin{proof}
At line 21, $I(d(u))$ must be less than $i$; otherwise, $u$ should have been extracted at Line 11 before extracting $u$ at Line 20, which is impossible. Also note that $d(u) \ge h(u)$, so $I(h(u))\leq I(d(u))$ for all $u \in V$. Thus, if $u_1 \in U_i$ is evaluated at Line 3 and there are $c - 1$ nodes $u_2, \dots, u_c$ extracted from $U_i$ at Line 21, then $I(d(u_1)), I(d(u_2)), \dots, I(d(u_c))$ should be in $c$ distinct ranges: $[I(h(u_1)), +\infty), [I(h(u_2)), I(h(u_1))), \dots, [I(h(u_c)), I(h(u_{c-1})))$, which implies that the number of different final values of $d'(u)$ in $U_i$ is at least $c$.

Suppose we remove all the cross-level edges in $T_i$, i.e., remove all the edges $(u, v)$ in $T_i$ whose final values of $d'(u)$ and $d'(v)$ differ. Then the tree should be decomposed into at least $c$ components since there are at least $c$ different final values of $d'(u)$ in $U_i$. Thus there are at least $c - 1$ cross-level edges in $T_i$.
\end{proof}

Finally, we can derive the $O(m + r \log k)$ time bound for our split algorithm.

\begin{proof}[Proof for Lemma \ref{lam:spl}]
By Lemma \ref{lam:qe}, the number of index evaluations at Line 13 is at most $r$. Let $c_i$ be the number of index evaluations at Line 3 and 21 for nodes in the group $U_i$. Then by Lemma \ref{lam:diffdp}, $\sum_{i = 1}^{b} c_i \le r + b$. Thus the total number of index evaluations in our split algorithm can be bounded by $2r + b \le O(r + n / \log l)$, which costs $O(r \log l + n)$ time.

At Line 3, Line 7-9, Line 21, we need to do a brute-force search in a given group, and each search costs $O(s) \le O(\log l)$ time. 
Note that the number of the brute-force searches can be bounded by twice the number of index evaluations at Line 3 and 21, so the total time cost for brute-force search is $O(r \log l + n)$.

It can be easily seen that all other parts of our split algorithm runs in linear time, so the overall time complexity is $O(m + r \log l) \le O(m + r \log k)$.
\end{proof}

\section{Discussion}

We give an improved algorithm for solving SSBP in $O(m\sqrt{\log n})$ time which is faster than sorting for sparse graphs. This algorithm is a breakthrough compared to traditional Dijkstra's algorithm using Fibonacci heap. There are some open questions remained to be solved. Can our algorithm for SSBP be further improved to $O(m \beta(m, n))$, which is the time complexity for the currently best algorithm for $s$-$t$ BP? Can we use our idea to obtain an algorithm for SSBP that runs faster than Dijkstra in word RAM model?

\bibliographystyle{alpha}
\bibliography{main}

\newcommand{\etalchar}[1]{$^{#1}$}
\begin{thebibliography}{CKT{\etalchar{+}}16}

\bibitem[BFP{\etalchar{+}}73]{blum1973time}
Manuel Blum, Robert~W. Floyd, Vaughan~R. Pratt, Ronald~L. Rivest, and
  Robert~Endre Tarjan.
\newblock Time bounds for selection.
\newblock {\em J. Comput. Syst. Sci.}, 7(4):448--461, 1973.

\bibitem[BKS02]{baier2002k}
Georg Baier, Ekkehard K{\"o}hler, and Martin Skutella.
\newblock On the k-splittable flow problem.
\newblock In {\em European Symposium on Algorithms}, pages 101--113. Springer,
  2002.

\bibitem[Bor]{Bor26}
O~Boruvka.
\newblock O jistem problemu minimalnim, praca moravske prirodovedecke
  spolecnosti 3, 1926.

\bibitem[Cha00]{Chaz00a}
B.~Chazelle.
\newblock A minimum spanning tree algorithm with inverse-{Ackermann} type
  complexity.
\newblock {\em J.~ACM}, 47(6):1028--1047, 2000.

\bibitem[CKT{\etalchar{+}}16]{chechik2016bottleneck}
Shiri Chechik, Haim Kaplan, Mikkel Thorup, Or~Zamir, and Uri Zwick.
\newblock Bottleneck paths and trees and deterministic graphical games.
\newblock In {\em LIPIcs-Leibniz International Proceedings in Informatics},
  volume~47. Schloss Dagstuhl-Leibniz-Zentrum fuer Informatik, 2016.

\bibitem[CT76]{CT76}
D.~Cheriton and R.~E. Tarjan.
\newblock Finding minimum spanning trees.
\newblock {\em SIAM J.~Comput.}, 5:724--742, 1976.

\bibitem[CW90]{Coppersmith1990}
Don Coppersmith and Shmuel Winograd.
\newblock Matrix multiplication via arithmetic progressions.
\newblock {\em Journal of Symbolic Computation}, 9(3):251 -- 280, 1990.
\newblock Computational algebraic complexity editorial.

\bibitem[Dij59]{dijkstra1959note}
Edsger~W Dijkstra.
\newblock A note on two problems in connexion with graphs.
\newblock {\em Numerische mathematik}, 1(1):269--271, 1959.

\bibitem[DP09]{duan2009fast}
Ran Duan and Seth Pettie.
\newblock Fast algorithms for (max, min)-matrix multiplication and bottleneck
  shortest paths.
\newblock In {\em Proceedings of the twentieth annual ACM-SIAM symposium on
  Discrete algorithms}, pages 384--391. Society for Industrial and Applied
  Mathematics, 2009.

\bibitem[EK72]{edmonds1972theoretical}
Jack Edmonds and Richard~M Karp.
\newblock Theoretical improvements in algorithmic efficiency for network flow
  problems.
\newblock {\em Journal of the ACM (JACM)}, 19(2):248--264, 1972.

\bibitem[FF56]{fordf1956}
L.~R. Ford and D.~R. Fulkerson.
\newblock {Maximal flow through a network}.
\newblock {\em Canadian Journal of Mathematics}, 8(0):399--404, January 1956.

\bibitem[FGA98]{Fernandez1998}
Elena Fernandez, Robert Garfinkel, and Roman Arbiol.
\newblock Mosaicking of aerial photographic maps via seams defined by
  bottleneck shortest paths.
\newblock {\em Oper. Res.}, 46(3):293--304, March 1998.

\bibitem[Fre85]{frederickson1985data}
Greg~N Frederickson.
\newblock Data structures for on-line updating of minimum spanning trees, with
  applications.
\newblock {\em SIAM Journal on Computing}, 14(4):781--798, 1985.

\bibitem[FT87]{fredman1987fibonacci}
Michael~L Fredman and Robert~Endre Tarjan.
\newblock Fibonacci heaps and their uses in improved network optimization
  algorithms.
\newblock {\em Journal of the ACM (JACM)}, 34(3):596--615, 1987.

\bibitem[GT88]{gabow1988algorithms}
Harold~N Gabow and Robert~E Tarjan.
\newblock Algorithms for two bottleneck optimization problems.
\newblock {\em Journal of Algorithms}, 9(3):411--417, 1988.

\bibitem[Hu61]{Hu1961}
T.~C. Hu.
\newblock The maximum capacity route problem.
\newblock {\em Operations Research}, 9(6):898--900, 1961.

\bibitem[Jar30]{Jar30}
Vojtech Jarn{\i}k.
\newblock O jist{\'e}m probl{\'e}mu minim{\'a}ln{\i}m.
\newblock {\em Pr{\'a}ca Moravsk{\'e} Pr{\i}rodovedeck{\'e} Spolecnosti},
  6:57--63, 1930.

\bibitem[KKT95]{karger1995randomized}
David~R Karger, Philip~N Klein, and Robert~E Tarjan.
\newblock A randomized linear-time algorithm to find minimum spanning trees.
\newblock {\em Journal of the ACM (JACM)}, 42(2):321--328, 1995.

\bibitem[Kru56]{Kru56}
Joseph~B Kruskal.
\newblock On the shortest spanning subtree of a graph and the traveling
  salesman problem.
\newblock {\em Proceedings of the American Mathematical society}, 7(1):48--50,
  1956.

\bibitem[PR02]{PR02c}
S.~Pettie and V.~Ramachandran.
\newblock An optimal minimum spanning tree algorithm.
\newblock {\em J.~ACM}, 49(1):16--34, 2002.

\bibitem[SYZ07]{shapira2007all}
Asaf Shapira, Raphael Yuster, and Uri Zwick.
\newblock All-pairs bottleneck paths in vertex weighted graphs.
\newblock In {\em Proceedings of the eighteenth annual ACM-SIAM symposium on
  Discrete algorithms}, pages 978--985. Society for Industrial and Applied
  Mathematics, 2007.

\bibitem[Tar72]{tarjan1972depth}
Robert Tarjan.
\newblock Depth-first search and linear graph algorithms.
\newblock {\em SIAM journal on computing}, 1(2):146--160, 1972.

\bibitem[VWY07]{vassilevska2007all}
Virginia Vassilevska, Ryan Williams, and Raphael Yuster.
\newblock All-pairs bottleneck paths for general graphs in truly sub-cubic
  time.
\newblock In {\em Proceedings of the thirty-ninth annual ACM symposium on
  Theory of computing}, pages 585--589. ACM, 2007.

\bibitem[Wil12]{Williams:2012}
Virginia~Vassilevska Williams.
\newblock Multiplying matrices faster than {Coppersmith-Winograd}.
\newblock In {\em Proceedings of the 44th symposium on Theory of Computing},
  STOC '12, pages 887--898, New York, NY, USA, 2012. ACM.

\bibitem[Yao75]{Yao75}
A.~C. Yao.
\newblock An {${O}(|{E}| \log \log |{V}|)$} algorithm for finding minimum
  spanning trees.
\newblock {\em Info.~Proc.~Lett.}, 4(1):21--23, 1975.

\end{thebibliography}

\appendix

\section{Tree Partition Algorithm} \label{app:partition}

Now we introduce the tree partition algorithm used in our \ModSSBPshort algorithm. There are many ways to do that. The tree partition algorithm we introduced here is a slight variant of the topological partition algorithm used in Frederickson's algorithm \cite{frederickson1985data}.

\paragraph*{Proof of Theorem~\ref{thm:partition}}
\begin{proof}
The core procedure in our tree partition algorithm is a recursive function \textproc{Partition} as shown in Figure \ref{fig:algopartition}. \textproc{Partition} takes a tree as the input, then it partitions the tree by calling itself for the subtrees. Using simple induction it can be shown that at any time from Line \ref{algopartition:line:Ustart} to \ref{algopartition:line:ret}, $U$ can always induce a connected subgraph in $T$, thus the induced subgraph is indeed a subtree.
Every time \textproc{Partition} reports a group $U_i$ at Line \ref{algopartition:line:report} during the running, we say that the subtree induced by $U_i$ in $T$ is chosen as a \textit{subtree candidate}.
The return value of \textproc{Partition} is a set of nodes $R$ which can induce a subtree in $T$ whose edges do not appear in any subtree candidate.

\begin{figure}[htbp]
\begin{algorithmic}[1]
\Function{Partition}{$T$}
    \State Let $v$ be the root of $T$
    \State $U \gets \{v\}$ \label{algopartition:line:Ustart}
    \ForAll {subtree $T'$ of $v$}
        \State $U \gets U \cup \Call{Partition}{T'}$
        \If {$\lvert U \rvert \ge s$}
            \State Report a new group $U$ \label{algopartition:line:report}
            \State $U \gets \{v\}$
        \EndIf
    \EndFor
    \State \Return $U$ \label{algopartition:line:ret}
\EndFunction
\end{algorithmic}
\caption{Tree Partition Algorithm} \label{fig:algopartition}
\end{figure}

To produce a tree partition for the spanning tree $T$ of $G$, we pass $T$ as the input to \Call{Partition}{$T$}. We collect the groups $U_1, \dots, U_b$ reported by \Call{Partition}{$T$} in order, and then merge the last group $U_b$ with the set of nodes $R$ returned by \Call{Partition}{$T$} or regard $R$ as a new group if no group has been reported before. For the former case, let $v$ be the root node $v$ when $U_b$ is reported, then $v$ must be contained in both $U_b$ and $R$, thus $U_b$ after the merge can still induce a subtree in $T$.

The running time of this algorithm is obviously $O(n)$. 
Now we turn to analyze the correctness.  
Since the subtrees are clearly edge-disjoint and every edge is contained in exactly one subtree, we only need to show that the size of each group is in $[s, 3s)$.
It is easy to see that the set of nodes returned at Line \ref{algopartition:line:ret} has the size in the range $[1, s]$, so the size of each subtree candidate is in the range $[s, 2s)$. If in the end there is no subtree candidate, then $\lvert R \rvert = s = n$ nodes will be returned and regarded as the only group in the partition. If there exist at least one subtree candidates, then the size of $U_b$ should be in the range $[s, 3s)$. Thus we complete our proof.
\end{proof}

\end{document}